\documentclass[runningheads]{llncs}

\usepackage[T1]{fontenc}
\usepackage{anyfontsize}
\usepackage{stmaryrd}
\SetSymbolFont{stmry}{bold}{U}{stmry}{m}{n} 
\usepackage
[disable]
{todonotes}
\usepackage{xspace}
\usepackage{physics} 
\usepackage{algpseudocode}
\usepackage{amssymb}
\usepackage{tikz} 
\usetikzlibrary{quantikz2}
\usetikzlibrary{shapes}
\usepackage{adjustbox}
\usepackage{graphicx} 
\usepackage[hypcap=false]{caption} 
\usepackage{bussproofs} 
\usepackage{subcaption}
\usepackage{listings}
\usepackage{xcolor} 
\usepackage{hyperref}
\hypersetup{
    breaklinks=true
}
\usepackage{booktabs}
\usepackage{framed} 
\usepackage{circuitikz}
\usepackage{array}
\usepackage{colortbl}
\definecolor{lightgreen}{RGB}{144, 238, 144}
\definecolor{lightred}{RGB}{255, 182, 193}
\usepackage{multirow}
\usepackage[normalem]{ulem}
\usepackage{etoolbox}
\apptocmd{\sloppy}{\hbadness 10000\relax}{}{}
\usepackage{bm} 
\newcommand{\myparagraph}[1]{\medskip \noindent \textbf{\sffamily #1.}}
\newcommand{\myaltparagraph}[1]{\medskip \noindent \textbf{\sffamily #1}}

\usepackage{wrapfig}
\usepackage{etoolbox}
\usepackage{pifont}

\newcommand{\xmark}[0]{\ensuremath{\times}\xspace}

\newcommand{\hcirc}[0]{Hybrid Circuit\xspace}

\newcommand{\ucirc}[0]{{Unitary Circuit}\xspace}
\newcommand{\qcirc}[0]{{Quantum Circuit}\xspace}

\newcommand{\dc}[0]{{\ensuremath{\mathsf{DC}}\xspace}}

\newcommand{\unit}[0]{\ensuremath{\mathsf{U}}\xspace}

\newcommand{\ium}{\ensuremath{\mathsf{IUM}}\xspace}
\newcommand{\iumo}{\ensuremath{\mathsf{I}_1\mathsf{U}_1\mathsf{M}_1}\xspace}
\newcommand{\iumt}{\ensuremath{\mathsf{I}_2\mathsf{U}_2\mathsf{M}_2}\xspace}

\newcommand{\meas}[0]{\ensuremath{\mathsf{Meas}}\xspace}
\newcommand{\init}[0]{\ensuremath{\mathsf{Init}}\xspace}

\newcommand{\inp}[0]{\ensuremath{\mathsf{InpVar}}\xspace}
\newcommand{\out}[0]{\ensuremath{\mathsf{Obs}}\xspace}

\newcommand{\dm}[0]{\ensuremath{\mathsf{{DM}}\xspace}}

\newcommand{\had}[0]{\ensuremath{\mathsf{{H}}}\xspace}
\newcommand{\xx}[0]{\ensuremath{\mathsf{{X}}\xspace}}
\newcommand{\zz}[1]{\ensuremath{\mathsf{{R_z(#1)}}\xspace}}
\newcommand{\uu}[0]{\ensuremath{\mathsf{{U1}}\xspace}}
\newcommand{\crz}[1]{\ensuremath{\mathsf{{CR_z(#1)}}\xspace}}

\newcommand{\cswap}[0]{\ensuremath{\mathsf{{CSWAP}}\xspace}}

\newcommand{\owm}[0]{{OWM\xspace}}
\newcommand{\tele}[0]{{Tele\xspace}}
\newcommand{\qiskit}[0]{{Qiskit\xspace}}

\newcommand{\qiskittr}[0]{\ensuremath{{{Qiskit_{tr}}}\xspace}}

\newcommand{\qft}[0]{\ensuremath{\mathsf{QFT}}\xspace}
\newcommand{\qpe}[0]{\ensuremath{\mathsf{QPE}}\xspace}
\newcommand{\shor}[0]{\ensuremath{\mathsf{Shor}}\xspace}

\newcommand{\udyn}[0]{\unit-\dc\xspace}
\newcommand{\hcqis}[0]{\dc-Q(\dc)\xspace}

\newcommand{\owmu}[0]{O(\unit)-\unit}

\newcommand{\teleu}[0]{T(\unit)-\unit}
\newcommand{\owmhc}[0]{O(\dc)-Q(\dc)\xspace}
\newcommand{\owmtele}[0]{O(\unit)-T(\unit)\xspace}

\newcommand{\openqasm}[0]{OpenQASM\xspace}

\newcommand{\pyzx}[0]{PyZX\xspace}

\newcommand{\qcec}[0]{QCEC\xspace}
\newcommand{\sliqec}[0]{SliQEC\xspace}

\newcommand{\ddiag}[0]{DD\xspace}

\newcommand{\autoq}[0]{AutoQ\xspace}
\newcommand{\autoqTwo}[0]{AutoQ-2\xspace}

\newcommand{\feynman}[0]{Feynman\xspace}

\newcommand{\sqbricks}[0]{SQbricks\xspace}
\newcommand{\sqbricksverif}[0]{SQbricks-Verif\xspace}
\newcommand{\sqv}[0]{SQV\xspace}

\newcommand{\sqbrickslift}[0]{SQbricks-Lifting\xspace}
\newcommand{\sql}[0]{SQL\xspace}

\newcommand{\veriqbench}[0]{VeriQBench\xspace}
\newcommand{\veriqc}[0]{VeriQC\xspace}
\newcommand{\qasmbench}[0]{QASMBench\xspace}
\newcommand{\cone}{\textsf{DisFree}\xspace}
\newcommand{\conetwo}{\textsf{Mix}\xspace}
\newcommand{\ctwo}{\textsf{Dis}\xspace}
\newcommand{\NA}{NA}
\newcommand{\NS}{NW}

\newcommand{\stand}[0]{standalone\xspace}
\newcommand{\lift}[0]{lifting\xspace}

\newcommand{\initqb}{{\bm{\shortmid}}\xspace}

\newcommand{\hequiv}{\ensuremath{\equiv_{\texttt{H}}}\xspace}

\newcommand{\G}{\ensuremath{\textsf{G}\xspace}}
\newcommand{\Ga}{\ensuremath{\textsf{GA}\xspace}}

\newcommand{\Ins}{\ensuremath{\textsf{Ins}\xspace}}
\newcommand{\Gp}{\ensuremath{\textsf{Ph}\xspace}}

\newcommand{\X}{\ensuremath{\textsf{X}\xspace}}
\newcommand{\Had}{\ensuremath{\textsf{H}\xspace}}
\newcommand{\wnot}{\ensuremath{\textsf{Not}\xspace}}
\newcommand{\qbit}{\ensuremath{\textsf{qb}\xspace}}
\newcommand{\cbit}{\ensuremath{\textsf{cb}\xspace}}
\newcommand{\Cir}{\ensuremath{\textsf{C}\xspace}}
\newcommand{\apply}{\ensuremath{\textsf{Apply}\xspace}}
\newcommand{\nil}{\ensuremath{\textsf{nill}\xspace}}

\newcommand{\linit}{\ensuremath{\texttt{lInit}\xspace}}
\newcommand{\emeas}{\ensuremath{\texttt{eMeas}\xspace}}

\newcommand{\headf}[2]{\ensuremath{#1-#2\xspace}}
\newcommand{\tailf}[2]{\ensuremath{#1-#2\xspace}}
\newcommand{\ifthen}[1]{\ensuremath{\textsf{if}\ #1 \ \textsf{then}\ }}

\DeclareRobustCommand{\ket}[1]{| #1 \rangle}

\begin{document}
\title{Quantum Circuit Equivalence Checking:\texorpdfstring{\\}{} A Tractable Bridge From Unitary to Hybrid Circuits}
\titlerunning{{Quantum Circuit Equivalence Checking}}
\author{Jérome Ricciardi\inst{1,2}\orcidID{0009-0001-7433-8384} \and
Sébastien Bardin\inst{1}\orcidID{0000-0002-6509-3506} \and
Christophe Chareton\inst{1}\orcidID{0000-0001-7113-563X} \and 
Benoît Valiron\inst{2}\orcidID{0000-0002-1008-5605}
}
\authorrunning{J. Ricciardi, et al.}

\institute{
Université Paris-Saclay, CEA, List, F-91120, Palaiseau, France, 
\email{first.name@cea.fr}\and
Université Paris-Saclay, CNRS, CentraleSupélec, ENS Paris-Saclay, Inria, Laboratoire Méthodes Formelles, 91190, Gif-sur-Yvette, France, \email{benoit.valiron@lmf.cnrs.fr}
}
\maketitle 
\begin{abstract}
  Equivalence checking of hybrid quantum circuits is of primary importance, given that quantum circuit transformations are omnipresent along the quantum compiler chain. 
  While some approaches exist for automating this
  task, most focus on the simple case of unitary circuits. At the same
  time, real quantum computing requires hybrid circuits equipped with
  measurement operators. Moreover, the few approaches targeting the
  hybrid case are limited to a restricted class of problems.  We propose
  tackling the Quantum Hybrid Circuit Equivalence Checking problem
  through lifting unitary circuit verification using a transformation
  known as deferred measurement. We show that this approach alone
  significantly outperforms prior work, and that, with the addition of
  specific unitary-level techniques we call separation and projection,
  it can handle much larger classes of hybrid circuit equivalence
  problems. We have implemented and evaluated our method over standard
  circuit transformations such as teleportation, one-way measurement,
  or the IBM Qiskit compiler, demonstrating its promises. As a side
  finding, we have identified and reported several unexpected behaviours with the Qiskit
  compiler.
\end{abstract}

\section{Introduction}

Quantum computation is a realm where information is stored on the state of objects governed by the laws of quantum physics \cite{nielsenChuang2002}. This model of computation is believed to provide important speedup for many applications, ranging from high-performance computing to optimisation. In recent years, quantum computers have become a near-term physical, industrial, and economic reality.

Compared to classical information, quantum data is very peculiar: it cannot be duplicated, and reading is a probabilistic operation done through \emph{measurement}, changing the global state of the memory. Unlike classical data, whose typical semantics is based on discrete data structures, quantum information is modelled with vectorial structures in Hilbert spaces. Because they manipulate data structures radically different from the classical case, quantum computers require specific developments at every development stage (user languages, verification, optimisation, compilation, etc.)~\cite{marco2023qverif,chareton2021formal}.

The typical execution flow for quantum programs relies on the notion of \emph{quantum coprocessor}: the quantum memory is stored in an external device, seen as a coprocessor to a CPU, similar to what happens for a GPU, for instance. The quantum coprocessor keeps the memory alive, while the main computation occurs on the classical CPU. A quantum process, therefore, consists of a (classical) interaction with the coprocessor by sending a series of instructions to initialise, update (with \emph{quantum, unitary gates}), and measure the quantum memory to retrieve classical information. A measurement's result is a classical piece of information that might be stored for later use or discarded. Such a series of instructions is represented by a \emph{quantum circuit}. A circuit can generally mix qubit initialisations, unitaries, measurements, and discards. We call \emph{unitary} (or \emph{purely quantum}) a circuit that consists only of quantum gates. A circuit with unitary gates, initialisation, measurements, {discards, and classically controlled instruction} is called a \emph{hybrid circuit}.

Known as the \emph{deferred measurement principle}, a result from folklore states that measurements and discards can be postponed to the end of the computation. This principle is partly why circuits found in most quantum algorithms are given uniquely in terms of unitary gates: measurement can always be thought of happening at the end of the computation. However, from an implementation point of view, it makes sense to consider the hybrid case as many quantum processes, such as repeat-until-success, measurement-based quantum computation, or optimisation techniques, rely on it.

Quantum circuit transformations turn a circuit into another (equivalent) quantum circuit. Such techniques are key components of the quantum software stack, whether for optimisation, adaptation to hardware capabilities in terms of qubits and connectivity, error correction, circuit robustification, one-way measurement. Compilers like $\qiskit$ do propose several transformation passes. As these passes are omnipresent between the programmer and the quantum hardware, checking their correctness is paramount. 
While we could envision quantum compilers fully certified in interactive theorem provers, another approach involves designing dedicated automated circuit equivalence checkers~\cite{marco2023qverif}.

This paper focuses on the automated equivalence verification problem arising from quantum circuit transformation: how to verify that two quantum circuits are functionally equivalent. In particular, we focus on verifying circuits involving initialisation, measurements, discards, and classically controlled instructions: hybrid circuits.

\myparagraph{Challenge with hybrid circuits equivalence verification} 
Current methodologies for automatic equivalence checking predominantly focus on purely quantum, unitary circuits~\cite{chen2025AutoQ20Unitary,amy2018towards,kissinger2020Pyzx,sander2024equivalencecheckingquantumcircuits}. This restriction is becoming increasingly unrealistic as physical chips progress. Indeed, the current trend in the design of quantum compilation toolchains tends to split the purely quantum processes into smaller components, to make them more robust to noise. Moreover, even pure quantum primitives require hybridisation and classical control for error detection and correction.

\medskip
Compared to purely quantum circuits, the formalisation of hybrid computations brings several extra difficulties.

A first difficulty is that hybrid quantum computation is, by nature, non-deterministic: the functional behaviour of a quantum program is a branching probabilistic structure instead of a linear trace. Quantum computation can then be regarded as a strict superset of probabilistic computation---a field where the formal analysis of programs is still a research question~\cite{Barthe_Katoen_Silva_2020}.
    
A second difficulty is that a quantum program manipulates both quantum and classical registers. Some of these registers might contain ``garbage'', i.e. data that should not be considered as output: a relevant equivalence notion should not consider them. We talk of \emph{partial equivalence} when the state equivalence is evaluated only up to non-discarded registers, an essential feature to model a program's state and its observable behaviour.
However, discarding quantum information is not innocuous, as it corresponds to tracing it out: discarding corresponds to a measurement and is then a probabilistic process.
 Therefore, the question of \emph{partial equivalence} of quantum programs is not trivial to define.
    
Finally, the measurement operation induces a deep conceptual shift from a (quantum) deterministic process purely modelled in linear algebra to a stochastic process acting on a mix of quantum states and usual data structures (Boolean values, integers, etc).
\begin{center}
  \textit{Designing a unified, tractable mechanism for the equivalence of hybrid quantum circuits is therefore a major challenge for formal methods.}
\end{center}
\myparagraph{Contributions}
To address this challenge, this paper presents an approach based on two main ingredients: (1) a classification of hybrid circuit equivalence problem instances, and (2) the lifting of verification methods designed for the unitary case to verify hybrid circuits, based on the deferred measurement principle mentioned earlier. In doing so, we advocate for the versatility of the formalism of \emph{path-sums} in the context of hybrid quantum circuit equivalence checking. More precisely, we bring the following contributions:
\begin{enumerate}
    \item
      We clarify the current landscape of hybrid circuit equivalence checking (Section~\ref{sec:overview}). In particular, we draw the separation between circuits with and without discards and the induced partitioning of equivalence cases.
    \item
      We propose a generic method based on deferred measurement that extends unitary circuit equivalence checkers to support \emph{hybrid circuits} (Section~\ref{se:partial-equivalences-of-circuits}). 
    \item
      We introduce a verification tool, \sqbricks, dedicated to hybrid circuit equivalence checking (Section~\ref{sec:implem}), included when discard is at stake. We systematically evaluate this implementation. 
\end{enumerate}
Overall, while our technique already allows for significant clarification and pushes forward the state-of-the-art of hybrid circuit equivalence checking, we also believe that our findings and benchmark establish a solid baseline for future research in the field.
To this end, \textit{\sqbricks and our experimental setup will be made available as open source.}

\myparagraph{State of the Art}\label{se:soa} 
Equivalence checking of purely unitary circuits is prevalent among the state-of-the-art. One can cite \autoq~\cite{chen2023AutoQ10Unitary,chen2024autoq20hybrid,chen2025AutoQ20Unitary}, \feynman~\cite{amy2018towards,amy2023complete,amy2025POPL}, \sliqec \cite{wei2022sliqec,wei2022sliqecPartialequiv}, and \pyzx~\cite{kissinger2020Pyzx}.  These tools leverage distinct methodologies, respectively: Tree Automata, Path-Sum, Decision Diagrams, and ZX-calculus.

On the other hand, to the best of our knowledge, only \qcec~\cite{burgholzerHandlingNonUnitariesQuantum2022} and \veriqc~\cite{xin2022DynamicEquivCheckTDD} address the equivalence checking of hybrid quantum circuits.  However, concrete implementations are limited to the case without discard (characterized as the \cone case in Section 4.5 of~\cite{xin2022DynamicEquivCheckTDD}) .

  \qcec relies on a combination of ZX-calculus, Decision Diagrams, Simulation, and deferred measurement for hybrid circuit equivalence.  Yet, the approach is limited \cite[Section.~4, p.~4]{burgholzerHandlingNonUnitariesQuantum2022} to the narrow case of hybrid circuits with no discards, preventing the analysis of typical transformations such as teleportation, one-way measurement, error correction, optimisations with ancillas, etc., and any form of non-reversible computation.

  \veriqc~\cite{xin2022DynamicEquivCheckTDD} implements measurement and classical control in quantum circuits using Tensor Decision Diagrams.  The authors have already identified one of the main challenges in hybrid circuit equivalence verification: managing discards, and they have formally defined an equivalence relation that accounts for it. Still, their tool does not tackle it and does not adress discard.
  
Table~\ref{tab:new-soa} provides an overview of the current state-of-the-art and our results.  It details the tools in terms of the technology they employ, the types of equivalence verification they support, and whether they rely on projection and separation techniques (last column).
\begin{table}[tb]
    \centering
    \footnotesize
    \caption{Automatic \qcirc Equivalence Tools. \sqv: \sqbricksverif,
  \sql: \sqbrickslift, (*): Any unitary equivalence checker considered
  here, **: The approach is unclear on these points (\veriqc), \(\dagger\): Meet our prerequisites
  (featuring separation tests and projections) to be able to verify
  lifted problems from \conetwo and \ctwo, \checkmark: Valid, \xmark:
  Invalid.}
    \label{tab:new-soa}
\begin{tabular}{ll|c|c|c@{\hspace{2mm}}c|c}
    \toprule
    &               & \multicolumn{4}{ c|}{\textbf{Equivalence}}   &{}\\
    &               & \multicolumn{4}{ c|}{\textbf{Checking Abilities}}   & Prerequisites$\dagger$\\
    {Name}          & {Techno.}     & Unitary       & \multicolumn{3}{ c|}{Hybrid} &{}\\
    &               &               & \cone & \conetwo  & \ctwo &   \\
    \autoq~\cite{chen2023AutoQ10Unitary,chen2025AutoQ20Unitary}          
    & Tree Automata & \checkmark    &    \xmark   & \xmark & \xmark& \xmark  \\
    \autoqTwo~\cite{chen2024autoq20hybrid}       
    & Tree Automata & \checkmark &\xmark&\xmark&\xmark&\xmark\\
    \feynman~\cite{amy2018towards,amy2023complete,amy2025POPL}        
    & Path-Sum & \checkmark &\xmark&\xmark&\xmark&\xmark\\
    \pyzx~\cite{kissinger2020Pyzx}           
    & ZX Calculus & \checkmark &\xmark&\xmark&\xmark&\xmark\\
    \sliqec~\cite{wei2022sliqec,wei2022sliqecPartialequiv}         
    & Decision Diagram (\ddiag)& \checkmark &\xmark&\xmark&\xmark&\xmark\\
    \midrule
    \qcec~\cite{sander2024equivalencecheckingquantumcircuits,Peham_2022,burgholzerHandlingNonUnitariesQuantum2022}           
    & ZX, DD, Simulation & \checkmark & \checkmark &\xmark&\xmark&\xmark\\
    \veriqc~\cite{xin2022EquivCheckTDD,xin2022DynamicEquivCheckTDD}         
    & Tensor DD & \checkmark & \checkmark & ** & ** & \xmark\\
    \midrule
    \midrule
    \sqv & Path-Sum & \checkmark &\xmark&\xmark&\xmark& \checkmark  \\
    \sqv + \sql  & Path-Sum & \checkmark & \checkmark & \checkmark & \checkmark & \checkmark  \\
    (*) + \sql & & \checkmark & \checkmark & \xmark & \xmark & \xmark \\
    \bottomrule
\end{tabular}
\end{table}

\section{Background}\label{sec:background}

This section is devoted to the presentation of the model of quantum computation, in particular, its main programming model: quantum circuits. We then introduce the two main aspects of the state-of-the-art that underpin our proposal: the \emph{deferred measurement principle} and the \emph{path-sums} semantics.

\myparagraph{Quantum Circuits}
The interaction with the quantum coprocessor consists of emitting a sequence of initialisation, unitary gates, measurements, and discards, potentially classically controlled by the result of previous measurements. 

\begin{example}\label{runexam}\rm 
  As an example of a hybrid circuit, Figure~\ref{fig:telep-intro} illustrates the \emph{teleportation protocol}~\cite{bennett1993teleporting}.  This protocol operates over mixed quantum and classical data and describes the transmission of a quantum state from Alice (input register \(\inp\)) to Bob (observable register \(\out\)).  The protocol highlights the general structure of quantum circuits: initialisation (\(\init\)), unitary evolution, measurement, discards, and classically controlled operations.
  The circuit consists of three quantum wires (single lines) and two classical wires (double lines).  It employs three types of gates: the Hadamard gate (\(\had\)), the (controlled) NOT gate \((\xx)\) and the phase gate \((\zz{1})\).
  Measurement \scalebox{.7}{\(\begin{quantikz} \meter{} \end{quantikz}\)} is a non-deterministic, described by the so-called  \emph{Born's rule}, with results stored in classical wires \scalebox{.7}{\(\begin{quantikz} \targ{} \end{quantikz}\)}. Classical wires are assumed to be initialised at $0$. Here, their values are not outputs of the circuit: they are discarded with \begin{quantikz}[wire types={c}]&\ground{}\end{quantikz}. 
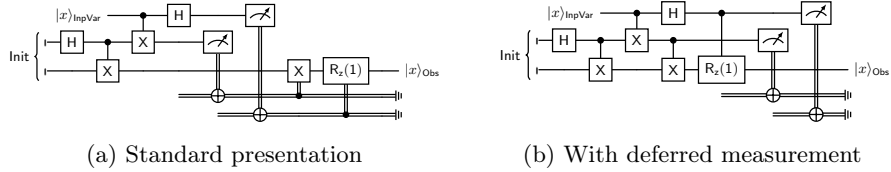
\begin{figure}[t] 
    \centering 
    \begin{subfigure}[b]{0.49\linewidth} 
    \centering 
    \adjustbox{width=\textwidth}{
    \begin{quantikz}[row sep=0.1cm,column sep=0.3cm,wire types={n,q,q,n,n}]
    &&\lstick{$\ket{x}_{\inp}$}&\ctrl{1}\setwiretype{q}&\gate{\had}&& \meter{} \wire[d][4]{c}\\ 
    \lstick[2]{$\init$}\initqb&\gate{\had}&\ctrl{1}&\gate{\xx}&&\meter{} \wire[d][2]{c}\\ 
    \initqb&&\gate{\xx}&&&&&\gate{\xx} \wire[d][1]{c} &\gate{\zz{1}} \wire[d][2]{c} & &\rstick{$\ket{x}_{\out}$}\\ 
    &&&&&\targ{}\setwiretype{c}&&\ctrl{0}&&&\ground{}\\ 
    &&&&&&\targ{}\setwiretype{c}&&\ctrl{0}&&\ground{}
    \end{quantikz} }
    \caption{Standard presentation}
    \label{fig:telep-intro}
    \end{subfigure} 
    \hfill 
    \begin{subfigure}[b]{0.49\linewidth} 
    \centering 
    \adjustbox{width=0.9\textwidth}{
    \begin{quantikz}[row sep=0.1cm,column sep=0.3cm,wire types={n,q,q,n,n}]
    &&\lstick{$\ket{x}_{\inp}$}&\ctrl{1}\setwiretype{q}&\gate{\had}&\ctrl{2}&& \meter{} \wire[d][4]{c}\\ 
    \lstick[2]{$\init$}\initqb&\gate{\had}&\ctrl{1}&\gate{\xx}&\ctrl{1}&& \meter{}\wire[d][2]{c}\\     \initqb&&\gate{\xx}&&\gate{\xx}&\gate{\zz{1}}&&&\rstick{$\ket{x}_{\out}$}\\
    &&&&&&\targ{}\setwiretype{c}&&\ground{}\\
    &&&&&&&\targ{}\setwiretype{c}&\ground{}
    \end{quantikz}}
    \caption{With deferred measurement}
    \label{fig:telep-dm}
    \end{subfigure} 
    \caption{The one-qubit teleportation algorithm}
    \label{fig:motiv-main}
\end{figure}
Historically, mathematical semantics for quantum circuits have relied on the underlying mathematical representation of unitary gates: linear, unitary maps. 
One problem is that the matrix size corresponds to the number of basis elements in the corresponding vector space: it is exponentially costly to compute. Another issue is that quantum circuits generally contain measurements, typically handled with more involved representations such as density matrices and superoperators. By sake of space, we cannot further introduce these quantum circuit component and their interpretation. We refer the desirous reader
to~\cite{nielsenChuang2002}~(Ch.4) for the standard introductory material.

\end{example}
\myaltparagraph{The deferred measurement}  is a circuit transformation that postpones all measurements to the end by replacing classical controls with quantum controls, leaving the inner circuit purely unitary.  
This transformation preserves the semantics, resulting in equivalent circuits.
Figure~\ref{fig:telep-dm} shows the result of applying the process to the teleportation algorithm of Figure~\ref{fig:telep-intro}. The output circuit consists of a round of initialisations \begin{quantikz}[wire types={q}]\initqb&\end{quantikz}, a unitary block, and a round of measurements, possibly immediately followed by discards (represented with \begin{quantikz}[wire types={q}]&\ground{}\end{quantikz}).
We refer to the resulting pattern as the \ium circuits. 
 Note that in our figures, the alignment of the wires in the drawing is not meaningful (they might, for instance, be shuffled inside the circuit).
 A formal definition of our implementation of the deferred measurement transformation is provided in Appendix~\ref{se:lift}.

\myparagraph{Path-sums}
An alternative to matrix representation has recently been proposed: \emph{path-sums}~\cite{amy2018towards,amy2023complete,vilmart2023rewriting,Vilmart2020SOP,deng2024case,chareton2021automated}.  This symbolic representation gives a more compact representation for purely quantum circuits, and has been at the core of an equivalence checking tool for unitary circuits: Feynman~\cite{amy2018towards,amy2023complete}. The name refers to \emph{Feynman's paths}, where the quantum evolution of a system is represented as a weighted sum of possible ``paths''. In the path-sums formalism, these weights are parameterised by the input basis states.

The path-sums formalism supports symbolic sequential (\cite[Definition~2.6]{amy2018towards}) and parallel composition of operators. 
Nevertheless, path-sums are over-expressive: two circuits corresponding to the same linear operation might have distinct path-sum representations. 
To reconcile such equivalence representations, path-sums come with an equational theory, made of (in addition to Boolean and dyadic ring theories) a set of rewriting rules simplifying path-sum while preserving $\mathcal{V}$ equivalence.

\section{Partial Equivalences of Quantum Circuits}
\label{sec:overview}\label{se:partial-equivalences-of-circuits}

The core of this paper focuses on the equivalence checking of quantum circuits. This section is dedicated to presenting our proposition's context and main structuring elements. 
For the sake of readability, in the following we assume the particular case where all input wires are quantum (the general case, with possibly classical inputs, is formally treated in Appendix~\ref{se:technical-details}). 
The prevalent approach to equivalence checking within the field predominantly addresses circuits with no discards.
This paper elaborates on our methodological approach to adapt and extend 
equivalence checking to encompass circuits with discards.

\myparagraph{Shapes of Hybrid Circuits}
We claim the presence or the absence of discards to be the main difficulty in the equivalence checking of hybrid circuits.
It straightforwardly generates a three classes typology of hybrid circuit equivalence instances: 
\begin{itemize}
\item
  \cone concerns circuits 
  without discards. The primary applications of \cone involve dynamising unitary circuits: basically, the transformation inverse of the deferred measurement, introducing measurement instructions in order to turn quantumly controlled commands into classically controlled ones. A key usage is robustification, with examples including the Quantum Fourier Transform \cite{bäumer2024quantumfouriertransformusing} and Quantum Phase Estimation \cite{corQPEDyn2021}. The state-of-the-art methods for hybrid circuits currently address confines to this case~\cite{burgholzerHandlingNonUnitariesQuantum2022,xin2022DynamicEquivCheckTDD}.
\item
  \ctwo, where both circuits feature
  discard.  \ctwo is the most generic equivalence task over hybrid circuits 
   teleportation, one-way measurement, error correction, optimisations with ancillas, etc.
   In fact, most of the non-reversible quantum processes require the use of ancillas and discards. 
\item
    \conetwo is the mixed case,  where one circuit is with discards, and the second is without.  With \conetwo, we can verify the equivalence between unitary and hybrid circuits resulting from transformations such as one-way measurement or teleportation. 
\end{itemize}

\myparagraph{Taking advantage of the Deferred Measurement Transformation}
To open the field to the equivalence checking of classes \conetwo and \ctwo, we propose a unified approach by capitalising on the \emph{deferred measurement principle} discussed in Section~\ref{sec:background}. 

Given a circuit $\Cir$, the deferred measurement transformation isolates the initialisation, unitary, and measurement components, obtaining a \ium circuit. 
A central result of this paper consists in characterising the relationship between the equivalence of two circuits $\Cir_1$ and $\Cir_2$ possibly with measures and discards, and the equivalence of their deferred measurement versions \(\iumo\) and \(\iumt\).
Indeed, it can be shown that 
if (i) $\mathsf{I_1}$ (resp. $\mathsf{M_1}$) and $\mathsf{I_2}$ (resp. $\mathsf{M_2}$) are equal and (ii) the unitary blocks $\mathsf{U}_1$ and $\mathsf{U}_2$ behaves equivalently over non-discarded qubits and
after initialisation (unitary circuit partial equivalence), one can
directly infer the hybrid equivalence between $\iumo$ and $\iumt$, and thus
address the original problem.

Hence, while our method pre-processes hybrid circuits along the deferred measurement transformation, it does not properly apply the principle. Instead, it builds a logical over-approximation of the equivalence relation that takes advantage of the deferred measurement transformation (the \ium circuit normal form) while eventually ignoring the final measurement.

\section{Implementation}
\label{sec:implem}
We have implemented our method in a prototype named \sqbricks, made of about 3,500 lines of OCaml code.  The code will be made open source\footnote{\url{https://github.com/Qbricks/qbricks.github.io/tree/main/Artifacts/SQbricks}}.  The implementation comprises two parts: (1) the \sqbricksverif component implements a path-sum calculus to perform unitary verification and partial equivalence checking; (2) the \sqbrickslift component focuses on our generic lifting method, based on the deferred measurement transformation.

\myaltparagraph{\sqbricksverif} 
(\sqv) is a unitary circuit verification tool based on path-sums, with a core rewriting system. Circuits are built from a  pseudo-universal set of gates, comprising \(\had\), \(\xx\), \(\uu\) (Z rotation up to global phase) and combined through sequence and quantum control.
Furthermore, the user can explicitly indicate the address of qubits that are discarded.

\myaltparagraph{\sqbrickslift}
(\sql) implements deferred measurement to extend unitary equivalence verification tools to hybrid circuits.
This approach maintains consistency between hybrid circuit classes (\cone, \conetwo, and \ctwo) while ensuring broad interoperability with existing verification tools through the \openqasm\ 2 standard.

\section{Experimental evaluation}\label{se:eval}\label{sec:eval}
To assess our circuit equivalence checking technique, we 
 conducted a series of benchmarks, all led on a laptop using Linux Mint 21.2, equipped with an Intel Core i9-9880H CPU (2.3GHz x 8), 31 GiB of RAM, and a 500 GB SSD.

\subsection{Experimental setup}\label{se:xp-setup}
\myparagraph{Considered tools}
We evaluate \sqbrickslift with the state-of-the-art unitary circuit verification tools \autoq (Tree Automata), \feynman~\cite{amy2018towards,feynmanGit} (Path-Sums), \pyzx~\cite{kissinger2020Pyzx,pyzxDoc} (ZX) and \qcec~\cite{burgholzerDD2020,burgholzerSim2020,qcecDoc} in unitary mode (ZX, Decision Diagram, Simulation, labelled \emph{\lift}, L in the following tables) in addition to \sqbricksverif (Path-Sums).
We compare our overall performance for hybrid equivalence checking against \qcec~\cite{burgholzerHandlingNonUnitariesQuantum2022} (in full hybrid mode, labelled \emph{\stand}, S)%
and \veriqc~\cite{xin2022DynamicEquivCheckTDD} (figures from Table 1 in~\cite{xin2022DynamicEquivCheckTDD}, since we were not able to use the tool directly (Appendix, Section~\ref{se:context-veriqc}). 

\myparagraph{Circuit collection}  
We selected two libraries of quantum circuits in \openqasm 2~\cite{cross2017openquantumassemblylanguage} commonly used in the literature---\veriqbench~\cite{chen2022veriqbenchbenchmarkmultipletypes} and the \feynman library~\cite{amy2018towards}, plus a hybrid implementation of the \shor algorithm~\cite{Shor_1997} (\hcirc) inspired by the version provided in \qasmbench~\cite{li2022qasmbenchlowlevelqasmbenchmark}.  For each library, we only retained the circuits that \sqbricks could parse (some seem to have issues or need refinement, see Appendix~\ref{sec:details-collec}).  Observe that multiple samples represent the same algorithms with varying parameter settings, typically the size.  This results in a compilation of 420 unitary circuits and 204 hybrid circuits.

\myparagraph{Circuit equivalence challenges}\label{pa:circ:eq:challenge}
We consider the following circuit transformations.
(1) $\mathbf{\qiskittr}$: the \texttt{generate\_preset\_path\_manager} of IBM \qiskit~\cite{javadiabhari2024quantumcomputingqiskit} with an optimisation level of 3 (the highest level) to maximise circuit transformations.
(2) $\mathbf{\owm}$: One-Way Measurement~\cite{danosMeasurementCalculus2007}, a circuit transformation aiming at enhancing the practical application of quantum computers by minimising the time a qubit remains coherent.
(3) $\mathbf{{\tele}}$: Teleportation, described in Figure~\ref{fig:telep-intro}.  
In the end, our equivalence challenge categories are: 
\begin{itemize}
    \item 
      \cone, with two classes of challenges:  169 pairs of unitary-hybrid circuits provided by \veriqbench and implementing the same algorithms (\udyn), either as a unitary or as a hybrid circuit; and  88 pairs consisting of a hybrid circuit and its transformed version using $\qiskittr$ (\hcqis).\label{item:cone}
    \item
      \conetwo, with three classes of challenges:  347 pairs made of a unitary circuit together with its \owm\ version (\owmu);  406 pairs made of a unitary circuit and its \tele\ version (\teleu);  88 pairs corresponding to a hybrid circuit transformed by both $\qiskittr$ and \owm\ (\owmhc).\label{item:cone-ctwo}
    \item \ctwo: 347 pairs (\owmtele) obtained from unitary circuits transformed with \owm\ and \tele.\label{item:ctwo}
\end{itemize}

\subsection{Experimental observations}\label{sec:expe_obs}
We address the three following research questions.
\begin{itemize}
\item RQ1: Can we lift unitary verification for checking \cone equivalence?
\item RQ2: How does it compare to prior works on \cone equivalence?
\item RQ3: Can we lift unitary verification for checking \ctwo and \conetwo equivalence?
\end{itemize}
We also performed a sanity check to evaluate the correctness of the considered tools (Section \ref{sec:sanity}), and report some additional findings about unexpected behaviours in $\qiskit$  we found along our experiments (Section \ref{sec:findings}).

\myaltparagraph{RQ1: Can we lift unitary  verification  for checking \cone?}
\label{su:se:rq1:class1}

For \cone challenges, we could draw performance comparisons between \sqv and our selected set of unitary verification tools.
To do so, we performed a pre-processing deferred measurement transformation over \cone circuits.
Results are summarized in  Table~\ref{tab:results}: 
path-sum based methods (\feynman and \sqv) achieved perfect success rates.  
In contrast, \pyzx demonstrated limitations, failing to verify equivalence between two \cswap\ variants (Appendix, Example~\ref{ex:pyzx-inconclusive-result}) and showing scalability issues.  For instance, \pyzx took 534s to verify a \hcqis \qpe instance of size 35, whereas \feynman verified an instance of size 42 in just 3s.  \autoq successfully verified all circuit families for small instances but faced scalability challenges, such as being limited to \hcqis to \qpe of size 6.

\textit{Conclusion:}
Our approach successfully lifts unitary verification tools to handle hybrid \cone cases, with path-sums showing superior performance compared to ZX-calculus and automata-based methods.
\begin{table}[tb]
    \centering
        \caption{Evaluation results of the lifting application of deferred
      measurement.
      {S: Standalone, L: Lifted},
      \dc: Dynamic Circuit, subclass of Hybrid Circuit
      without Discard, O: \owm,\quad T: \tele,\quad Q: \qiskittr, TO:
      Time Out (10 min) Wrong: equivalence check returns not
      equivalent for  equivalent circuits, \NA: Not
      Applicable, \NS: Not Working.}
    \label{tab:results}
    \resizebox{\textwidth}{!}{    
    \begin{tabular}{lc|c|c|c|c|c|c|c|c}
        \midrule
        &&\multicolumn{7}{c|}{Success} & \\
        \multicolumn{2}{c|}{ } & \multicolumn{2}{c|}{\cone} & \multicolumn{3}{c|}{\conetwo} & \ctwo & \\ 
        Tool& Lift& \udyn & \hcqis & \owmu & \teleu & \owmhc & \owmtele & Total & TO \\ 
        \midrule
        \multirow{2}{*}{QCEC} & \multirow{2}{*}{S} & 0 & \multirow{2}{*}{87} 
        & \multirow{2}{*}{\NA} & \multirow{2}{*}{\NA} & \multirow{2}{*}{\NA} & \multirow{2}{*}{\NA} &  87 & \multirow{2}{*}{143} \\ 
        &  & Wrong: 27 &  &  &  &  &  & Wrong: 27 \\ 
        VeriQC (Cf. Table~\ref{tab:veriqc-bench}) & S & \NS & \NS & \NA & \NA & \NA & \NA & \NA & \NA \\
        \midrule
        AutoQ-2.0 & L &10& 15 & \NA &\NA  &\NA  &\NA  & 25 & 23 \\

        Feyn-24 & L &169& 88 & \NA &\NA  &\NA  &\NA  &  257 & 0 \\

        PyZX & L &73& 66 & \NA &\NA  &\NA  &\NA  & 139 &112  \\

        QCEC & L & 151 & 69 & \NA & \NA & \NA & \NA & 220 & 19  \\

        \midrule

        {\sqv} & L & 169 & {88} & {180} & {343} & {50} & {137} & {967} & {426} \\

        \midrule
        \multicolumn{2}{c|}{\#challenges} & 169 & 88 & 347 & 406 & 88 & 347&
        \multicolumn{1}{c}{1445} &\\ 
        \bottomrule 
    \end{tabular}
    }
\end{table}

\myaltparagraph{RQ2: How does it compare to  prior works on \cone?}\label{su:se:rq2:class1}
We compare \qcec~\cite{burgholzerHandlingNonUnitariesQuantum2022} and \veriqc~\cite{xin2022DynamicEquivCheckTDD} for \cone\ hybrid equivalence, evaluating \qcec\ directly and \veriqc\ via published results.
\qcec\ (standalone), lifted \feynman, and lifted \sqv\ show similar performance on 88 \hcqis\ challenges, but \qcec\ fails all 169 \udyn\ challenges, and incorrectly reports 27 non-equivalence proofs, indicating correctness issues (see Appendix, Example~\ref{ex:qcec-wrong-result} for a minimal example).
\veriqc \cite{xin2022DynamicEquivCheckTDD} performs hybrid equivalence checking on \cone.
\footnote{The paper claims results over \ctwo \cite[Definition 2]{xin2022DynamicEquivCheckTDD}. However,~\cite[Section 4.5]{xin2022DynamicEquivCheckTDD} explains how to check dynamic circuit equivalence with TDDs, but is limited to circuits without discard (our \cone).}
Unfortunately, despite contacting \veriqc~\cite{xin2022DynamicEquivCheckTDD}  authors, we couldn't use this tool for our experiments.  We compared our methods against its published~\cite{xin2022DynamicEquivCheckTDD} experimental performance data.
 Results are shown in Table~\ref{tab:veriqc-bench}.
With our lifting approach, all tools except \autoq verify the entire \veriqc benchmark faster than \veriqc itself.  

\medskip
\textit{Conclusion:}
Our lifting method is highly effective on \cone challenges.  While our method matches \qcec's performance on its best subcategory, 
it outperforms both \qcec and \veriqc in all other cases.
\begin{table}[t]
    \centering 
        \caption{
    Comparing against \veriqc from published results on their available benchmark (selection of the most significant results, time in seconds, \checkmark: Success)
    }
    \label{tab:veriqc-bench}
    \resizebox{.7\textwidth}{!}{    
    \begin{tabular}{lc|ccccc}
    \toprule
    \multicolumn{2}{c|}{}&\multicolumn{5}{c}{\textbf{\sqbrickslift} + }\\
    \midrule
    \textbf{Task} 
    & \textbf{\veriqc\cite{xin2022DynamicEquivCheckTDD}} 
    \hspace{1mm} & \hspace{1mm}\textbf{\autoq} 
    & \textbf{Feyn-24} 
    & \textbf{\pyzx} 
    & \textbf{\qcec} 
    & {\textbf{\sqv}}  \\
    \midrule
    \qft\_11  & 0.86 \checkmark & Err & 0,01 \checkmark& 0,39 \checkmark& 0,01 \checkmark& 0,01 \checkmark \\
    
    \qft\_16  & 23.38 \checkmark&Err& 0,02 \checkmark& 3,64 \checkmark& 0,02 \checkmark& 0,03 \checkmark \\
    pe\_9    & 1.62 \checkmark&Err& 0,01 \checkmark& 0,29 \checkmark & 0,01 \checkmark& 0,01 \checkmark \\
    phaseflip & 0.18 \checkmark& 2,50 \checkmark& 0,00 \checkmark& 0,01 \checkmark& 0,00 \checkmark& 0,00 \checkmark \\
    teleportation & 0.01 \checkmark& 0,00 \checkmark& 0,00 \checkmark & 0,00 \checkmark& 0,00 \checkmark& 0,00  \checkmark\\
    state\_inj\_T & 0.01 \checkmark& 0,00 \checkmark& 0,00 \checkmark& 0,00 \checkmark& 0,00 \checkmark& 0,00 \checkmark\\
    \bottomrule
    \end{tabular}
}
\end{table}

\myaltparagraph{RQ3: Can we lift  unitary  verification for checking \ctwo and \conetwo?}\label{su:se:rq3:class2}

This problem considers challenges from \ctwo and \conetwo, out of the scope of prior work.
Therefore, 
no comparative analysis was possible here.
Our method effectively addresses a substantial portion of the \conetwo\ and \ctwo\ hybrid equivalence challenges, handling 39.5\%
of \ctwo\ and 
67.7\% of \conetwo. 

\medskip

\textit{Conclusion:} 
Our approach can indeed address hybrid equivalence checking out of the scope of the current state-of-the-art tools, with a reasonable success rate, establishing an acceptable first solution for these problems.

\subsection{Additional Findings} \label{sec:findings}
\myparagraph{Sanity check} \label{sec:sanity}
We also performed a sanity check consisting of 73 equivalence tasks from \veriqbench\ and \qasmbench, with deliberately modified quantum circuits (mutants).  These modifications ensure non-equivalence by design.  
All versions of \feynman\ and \sqbricks\ successfully passed the sanity check with no false positives.  Other tools exhibited two types of failures, primarily related to rotation gates:
(1) \pyzx\ and \qcec\ failed to handle very small angles ($\leq \pi/2^{26}$ and $\leq \pi/2^{27}$ respectively).
(2) \autoq\ failed to distinguish between controlled and uncontrolled $Z$ axis rotation gates (\crz{k} and \zz{k} gates).

\myparagraph{Unexpected behaviors with $\qiskittr$} 
During our experiments, we uncovered two bugs in the \qiskit\ compiler:
(1) \textbf{Angle approximation:} \qiskit\ version 1.1.0 approximated small angles to 0, leading to incorrect circuit simplifications in e.g., \qft, \qpe, \shor algorithm, etc, when quantum registers are over 42 qbits\footnote{This simplification imposed the induced restrictions over our experiments}.  This issue was fixed in version 1.4.0.
(2) \textbf{Introduction of Floats:} In version 1.4.0, transformations introduced floating-point values (e.g., converting $3 \pi / 32$ to $0.2945243112740431$) instead of rational numbers, causing equivalence loss in some circuits. We reported this to the \qiskit\ team.
These findings highlight the practical utility of our method in identifying and mitigating approximation-related issues in quantum circuit transformations.

\myparagraph{Acknowledgments}
This work has been partially funded by the French National Research Agency
(ANR): projects TaQC
ANR-22-CE47-0012 and within the framework of ``Plan France 2030'',
under the research projects EPIQ ANR-22-PETQ-0007, OQULUS
ANR-23-PETQ-0013, HQI-Acquisition ANR-22-PNCQ-0001 and HQI-R\&D
ANR-22-PNCQ-0002.

\appendix 

\section{Experimental evaluation details}\label{se:Experimental evaluation details}

This appendix provides details about our experimental evaluation, including specifics about the collection of circuits we used (Appendix~\ref{sec:details-collec}) 
and information about the tools we compared ourselves to (Appendix~\ref{sec:tool-limitations}).

\subsection{Details on Circuit Collection}\label{sec:details-collec}
Our circuit collection comprises 
two categories:
\begin{itemize} 

    \item \textbf{{\ucirc}s}: 
    420  circuits, including 43/44 circuits from the \feynman library
    \footnote{Except \lstinline{cycle_17_3.qasm} which has an implementation issue}.
    and 377/782 circuits from the \veriqbench library \lstinline{combinational} subset, and without the \lstinline{sequential} and \lstinline{variational} subsets due to parsing issues.
    Parsing issues arise if a circuit is ill-formed or contains elements not accounted for in our syntax, such as macros.  
    
    \item \textbf{{\hcirc}s}: 
    204 circuits, including 198/205  circuits from the \veriqbench library (\qpe and \qft, bit flip and phase flip correction, state injection and teleportation) and an implementation of \shor~\cite{Shor_1997} over 5 qubits inspired from the \qasmbench~\cite{li2022qasmbenchlowlevelqasmbenchmark} library. 
\end{itemize}

\subsection{Tool Limitations and Failure Cases}\label{sec:tool-limitations}

This appendix provides information on the benchmark for \veriqc (Table~\ref{tab:veriqc-bench}), and
presents concrete examples demonstrating specific limitations of state-of-the-art 
quantum circuit verification tools, including cases where tools return incorrect 
or inconclusive results despite the functional equivalence of the circuits being compared.

\noindent\paragraph{\veriqc.}\label{se:context-veriqc}
After encountering difficulties in reproducing their experiments with \cone, we began email correspondence with the authors. However, we faced ongoing challenges in replicating their results. As a result, we utilised their findings: \cite[Table 1]{xin2022DynamicEquivCheckTDD}.

\begin{example}[\pyzx minimal inconclusive result]\label{ex:pyzx-inconclusive-result}
    The following pair of circuits implements the controlled-swap operation (\cswap) in two different ways. 
    This operation swaps two qubits conditionally on the state of a third control qubit. 
    For basis states $|x_0,x_1,x_2\rangle$, the operation should:
    \begin{itemize}
        \item Leave $x_0$ unchanged (control qubit)
        \item Swap $x_1$ and $x_2$ when $x_0 = 1$
        \item Leave $x_1$ and $x_2$ unchanged when $x_0 = 0$
    \end{itemize}
    While these implementations are functionally equivalent (both implement controlled-swap), \pyzx~\cite{kissinger2020Pyzx} {(version 0.9.0)} fails to verify this equivalence, demonstrating a limitation in handling certain control structures.

    \scalebox{.9}{\begin{minipage}{0.52\textwidth}
        \centering
        
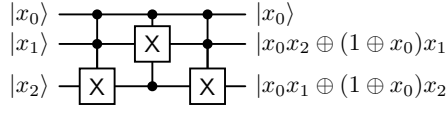
\captionof{figure}{\cswap\ first implementation}
        \begin{quantikz}[row sep=0.1cm, column sep=0.3cm, wire types={q,q,q}]
            \lstick{$\ket{x_0}$} & \ctrl{2} & \ctrl{1} & \ctrl{2} & \rstick{$\ket{x_0}$} \\
            \lstick{$\ket{x_1}$} & \ctrl{1} & \gate{\xx} & \ctrl{1} & \rstick{$\ket{x_0x_2 \oplus (1 \oplus x_0) x_1}$} \\
            \lstick{$\ket{x_2}$} & \gate{\xx} & \ctrl{-1} & \gate{\xx} & \rstick{$\ket{x_0x_1 \oplus (1 \oplus x_0) x_2}$}
        \end{quantikz}
    \end{minipage}}
    \scalebox{.9}{\begin{minipage}{0.52\textwidth}
        \centering
        
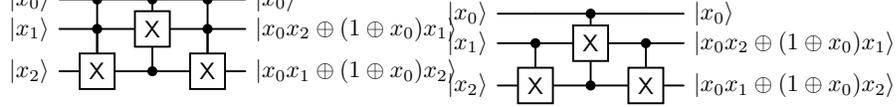
\captionof{figure}{\cswap\ alternative implementation}
        \begin{quantikz}[row sep=0.1cm, column sep=0.3cm, wire types={q,q,q}]
            \lstick{$\ket{x_0}$} & & \ctrl{1} & & \rstick{$\ket{x_0}$} \\
            \lstick{$\ket{x_1}$} & \ctrl{1} & \gate{\xx} & \ctrl{1} & \rstick{$\ket{x_0x_2 \oplus (1 \oplus x_0) x_1}$} \\
            \lstick{$\ket{x_2}$} & \gate{\xx} & \ctrl{-1} & \gate{\xx} & \rstick{$\ket{x_0x_1 \oplus (1 \oplus x_0) x_2}$}
        \end{quantikz}
    \end{minipage}}
\end{example}

\begin{example}[\qcec minimal incorrect result]\label{ex:qcec-wrong-result}
    The circuits below illustrate a minimal case where \qcec~\cite{burgholzerHandlingNonUnitariesQuantum2022} {(version 2.8.1)}, 
    when using its deferred measurement option, 
    incorrectly determines that the circuits $\Cir_1$ and $\Cir_2$ are not equivalent. 

    \scalebox{.9}{\begin{minipage}{0.52\textwidth}
        \centering
        \captionof{figure}{$\Cir_1$}
           \begin{quantikz}[row sep=0.1cm, column sep=0.3cm, wire types={q,q}]
            & \meter{} \wire[d][1]{c} \\
        & \gate{\zz{1}} & \gate{\had} & \meter{}
    \end{quantikz}
    \end{minipage}}
    \scalebox{.9}{\begin{minipage}{0.52\textwidth}
        \centering
        \captionof{figure}{$\Cir_2$}
        \begin{quantikz}[row sep=0.1cm, column sep=0.3cm, wire types={q,q}]
            & \ctrl{1} & & \meter{} \\
        & \gate{\zz{1}} & \gate{\had} & \meter{}
    \end{quantikz}
    \end{minipage}}
\end{example}

\section{Technical Details}\label{se:technical-details}
 
\subsection{Lifting unitary verification tools for hybrid circuits}\label{se:lift}

This section provides the main formal ingredients for our deferred measurement-based lifting of unitary circuits' equivalence to the hybrid case.

\begin{definition}[Hybrid circuits]\label{def:hyb-cir}
    A hybrid circuit is a sequence of instructions generated by the following syntax
    \[
    \begin{array}{rcl}
    \G     &:=& \Gp(k)\mid \zz{k}\mid \X\mid \Had  \\
    \Ga     &:=& \apply (\G, [\qbit],[\qbit])  \\
    \Ins     &:=&  \Ga \mid\ifthen{ \cbit} \Ga \mid \meas (\qbit,\cbit)\mid  \init(\qbit) \mid\wnot (\cbit)\\
    \Cir     &:=& [\Ins]
    \\
    \end{array}
    \]
    where $[t]$ denotes a standard list construction of type $t$, built either  as the nill list $[l]:= \nil$ or as a $t$ type object $o$ appended to  another $t$ list  $[l]:= \headf{o}{[l']}$. 
    By abuse of notation, we use the same $\headf{l}{l'}$ for $l,l'$ being either a list or a type $t$ object (assimilated to a list of one single element). 
    
     $\apply (\G, [\qbit_1],[\qbit_2])$ intuitively commands the parallel application of gate $\G$ on qubits in $[\qbit_2]$, controlled by the conjunction of qubits in $[\qbit_1]$; $\meas (\qbit,\cbit)$ commands the measure of qubit \qbit\ and the storage of the obtained data in classical bit \cbit
    A  circuit is \emph{well-formed} if it respects the following syntactical constraints: that (i) a classical wire should receive at most one measurement result in a given circuit and (ii) a quantum wire is not further addressed when having been measured.
    \label{well-formedness}
\end{definition}

The deferred measurement principle is introduced in~\cite{nielsenChuang2002} as a \emph{rather obvious} property of circuits: that measurement can always be moved from an intermediate stage of the computation to the end of the circuit, while classically controlled instructions are replaced by quantum conditionals. 
Surprisingly, to the best of our knowledge, the first formal and generic proof of the principle was established as late as 2022~\cite{Devitt_2013}, and a computer-assisted proof was even recently given in~\cite{ying2024lawsquantumprogramming}.

 For a circuit $\Cir$, we write \linit(\Cir) the sum of, for each initialisation command $\init$, the number of non-initialisation commands that precede it in $\Cir$. Symmetrically, \emeas(C) denotes the sum of, for each measure command \meas in $\Cir$, the number of non-measure commands succeeding it in $\Cir$.
We also informally introduce the function $m_{\Cir}$. 
It maps a classical wire $\cbit$ to a quantum wire  $\qbit$ in instructions. In our deferred measurement circuit transformation, 
as the measurement  $\meas(\qbit,\cbit)$ of a qubit is postponed to the end of the execution, the intermediary classical control commands over $\cbit$ should be turned into quantum control commands over $\qbit$. This is achieved by applying the instruction transformation $m_{\Cir}(\qbit,\cbit)$ to the instructions $\meas(\qbit,\cbit)$ is permuted with. In addition, 
 it turns a classical bit-flip instruction  $\wnot(\cbit)$ into its quantum counterpart $\apply (\X, \nil,\qbit)$. Formally, for any instruction \Ins, we have that:
\[\begin{array}{rcl}
\Ins[ m_{\Cir}(\cbit,\qbit)] &=& \textbf{if } \Ins := \ifthen{ \cbit} \apply (\G, [co],[ta])~
\textbf{then } 
 \apply (\G, [co\cup\{\qbit\}],[ta])  \\
&&\textbf{else  if }\Ins := \wnot(\cbit)\\
&&\qquad \textbf{then }\apply (\X, \nil,\qbit)\\
 && \qquad \textbf{else }\Ins
\end{array}
 \]
 
\begin{definition}[Deferred measurement circuit transformation]\label{def:dm}
     The deferred measurement transformation is built as a double inductive rewriting pass, where initialisation (resp. measurement) commands systematically commute with non-init (non-measure) commands whenever they occur in non-initial (resp. non-final) position.
    {\footnotesize\[
    \begin{array}{rlllll}
        \texttt{pI(\Cir)} :=& \textit{if }& \linit(C) = 0 &\textit{then}& \Cir& \\
        &\textit{ else if}  & \Cir:= \headf{\init(\qbit)}{\Cir'} &\textit{then}&\headf{\init(\qbit)}{\texttt{pI}(\Cir')}&  \\
        & \textit{ else if} & \Cir:= \headf{\Ins}{\headf{\init(\qbit)}{\Cir'}} &\textit{then}& \headf{\init(\qbit)}{\texttt{pI}(\headf{\Ins}{\Cir'}})&  \\
        &\textit{ else if}  & \Cir:= \headf{\Ins}{\Cir'}  &\textit{then}&\texttt{pI}(\headf{\Ins}{\texttt{pI}(\Cir')}) \\ 
        &&&&& \\\texttt{lM}(\Cir):= & \textit{if }& \texttt{lM}(\Cir) = 0 &\textit{then}& \Cir&\\
        & \textit{ else if} & \Cir:= \tailf{\Cir'}{\meas(\qbit,\cbit)} &\textit{then}&\tailf{\texttt{lM}(\Cir')}{\meas(\qbit)}& \\
        & \textit{ else if} & \Cir:= \tailf{\tailf{\Cir'}{\meas(\qbit,\cbit)}}{\Ins} &\textit{then}& \tailf{\texttt{lM}(\tailf{\Cir'}{\Ins[ m_{\Cir}(\qbit,\qbit)]})}{\meas(\qbit,\cbit)}&  \\
        & \textit{ else if} & \Cir:= \tailf{\Cir'}{\Ins}  &\textit{then}&\texttt{lM}(\tailf{\texttt{lM}(\Cir')}{\Ins}) \\
        &&&&& \\\dm(\Cir)&:=&\texttt{lM}(\texttt{pI}(\Cir))      
    \end{array}
    \]}
\end{definition}

This transformation provides a circuit  equivalent to $\Cir$
and sequentially structured as three successive blocks $\headf{\headf{[\init]_{\Cir}}{[\texttt{U}]_{\Cir}}}{[\meas]_{\Cir}}$ of (i) initialisation  (ii) unitary gate application, and (iii) measure commands. Formally, we have the following theorem:

\begin{theorem}[Deferred measurement]\label{thm:dm}
Let $\Cir$ be a hybrid circuit, then:
(i) $\dm(\Cir) \hequiv \Cir$, (ii) $\linit(\dm(\Cir)) = 0$,
(iii) $\emeas(\dm(\Cir)) = 0$
\end{theorem}

\begin{proof}[Sketch]
    By structural induction over $\Cir$. Transformation \texttt{pI} preserves the semantical equivalence and ensures condition 2., transformation \texttt{lM} preserves both and ensures condition 3 in addition.
\end{proof}

\begin{example}\label{ex:tele}
As an illustration of Definition~\ref{def:dm}, Figure~\ref{fig:telep-dm} draws it application to the teleportation case discussed in Example~\ref{runexam}: measurement instructions are delayed to the end of the execution, and intermediary classical instruction -- originally controlling over measurement results -- are turned into quantum controlled instruction -- controlling over the corresponding \emph{not yet measured} quantum wires.
\end{example}


\begin{thebibliography}{10}
\providecommand{\url}[1]{\texttt{#1}}
\providecommand{\urlprefix}{URL }
\providecommand{\doi}[1]{https://doi.org/#1}

\bibitem{chen2025AutoQ20Unitary}
Abdulla, P.A. \textit{et al.}: Verifying quantum circuits with
  level-synchronized tree automata. Proc. POPL \textbf{9} (2025). 

\bibitem{amy2018towards}
Amy, M.: Towards Large-scale Functional Verification of Universal Quantum
  Circuits. EPTCS \textbf{287} (2019).

\bibitem{amy2023complete}
Amy, M.: Complete equational theories for the sum-over-paths with unbalanced
  amplitudes. EPTCS 
  \textbf{384} (2023).

\bibitem{amy2025POPL}
Amy, M., Lunderville, J.: Linear and non-linear relational analyses for quantum
  program optimization. Proc. POPL \textbf{9} (2025).

\bibitem{autoqGit}
Auto{Q} {G}it, \url{https://github.com/fmlab-iis/AutoQ}

\bibitem{Barthe_Katoen_Silva_2020}
Barthe, G., Katoen, J.P., Silva, A. (eds.): Foundations of Probabilistic
  Programming. Cambridge University Press (2020)

\bibitem{bennett1993teleporting}
Bennett, C.H. \text{et al.}: Teleporting an unknown quantum state via dual classical and
  Einstein-Podolsky-Rosen channels. Phys. Rev. Lett.  \textbf{70},  1895--1899
  (1993).

\bibitem{burgholzerDD2020}
Burgholzer, L., Wille, R.: Improved DD-based equivalence checking of quantum
  circuits. In: Proc. ASP-DAC (2020).

\bibitem{burgholzerSim2020}
Burgholzer, L., Wille, R.: The power of simulation for equivalence checking in
  quantum computing. In: Proc. DAC (2020).

\bibitem{burgholzerHandlingNonUnitariesQuantum2022}
Burgholzer, L., Wille, R.: Handling {{Non-Unitaries}} in {{Quantum Circuit
  Equivalence Checking}}. In: Proc. DAC (2022).

\bibitem{bäumer2024quantumfouriertransformusing}
Bäumer, E. \textit{et al.}: Quantum Fourier
  transform using dynamic circuits (2024).

\bibitem{chareton2021automated}
Chareton, C. \textit{et al.}: An automated
  deductive verification framework for circuit-building quantum programs. ESOP (2021).

\bibitem{chareton2021formal} Chareton, C. \textit{et al.}: Formal
  methods for quantum algorithms. In: Handbook of Formal Analysis and
  Verification in Cryptography. {CRC} Press (2023).

\bibitem{chen2022veriqbenchbenchmarkmultipletypes} Chen, K. \textit{et
    al.}: Veriqbench: A benchmark for multiple types of quantum
  circuits. arXiv:2206.10880 (2022).

\bibitem{wei2022sliqecPartialequiv} Chen, T.F. \textit{et al.}:
  Partial equivalence checking of quantum circuits. QEC
  2022.

\bibitem{chen2024autoq20hybrid}
Chen, Y.F. \textit{et al.}: AutoQ 2.0: From verification of quantum circuits to verification
  of quantum programs. arXiv:2411.09121 (2024).

\bibitem{chen2023AutoQ10Unitary}
Chen, Y.F. \textit{et al.}: AutoQ: An
  automata-based quantum circuit verifier. CAD 2023.

\bibitem{corQPEDyn2021}
C\'orcoles, A.D. \textit{et al.}: Exploiting dynamic quantum circuits in a quantum algorithm
  with superconducting qubits. Phys. Rev. Lett.  \textbf{127},  100501 (2021).

\bibitem{cross2017openquantumassemblylanguage}
Cross, A.W., Bishop, L.S., Smolin, J.A., Gambetta, J.M.: Open quantum assembly
  language. arXiv:1707.03429 (2017).

\bibitem{danosMeasurementCalculus2007}
Danos, V., Kashefi, E., Panangaden, P.: The {{Measurement Calculus}}. (2007)

\bibitem{deng2024case}
Deng, H., Tao, R., Peng, Y., Wu, X.: A case for synthesis of recursive quantum
  unitary programs. Proc. POPL (2024).

\bibitem{feynmanGit}
Feynman {G}it, \url{https://github.com/meamy/feynman}

\bibitem{Devitt_2013}
Devitt, S. J. and Munro, W. J. and Nemoto, K.: Quantum error correction for beginners. (2013),

\bibitem{xin2022DynamicEquivCheckTDD}
Hong, X., Feng, Y., Li, S., Ying, M.: Equivalence checking of dynamic quantum
  circuits. In: Proc. ICCAD (2022)

\bibitem{xin2022EquivCheckTDD}
Hong, X. \textit{et al.}: A tensor network based decision
  diagram for representation of quantum circuits. ACM Trans. Des. Autom.
  Electron. Syst.  \textbf{27}(6) (Jun 2022).

\bibitem{ying2024lawsquantumprogramming}
Ying M. and Zhou L. and Barthe G.: Laws of Quantum Programming. (2024).

\bibitem{ioannou2006computational}
Ioannou, L. M.: Computational complexity of the quantum separability problem (2006).

\bibitem{javadiabhari2024quantumcomputingqiskit}
Javadi, A. \textit{et al.}: Quantum computing with Qiskit. arXiv:2405.08810 (2024).

\bibitem{kissinger2020Pyzx}
Kissinger, A., van~de Wetering, J.: {PyZX: Large Scale Automated Diagrammatic
  Reasoning}. In: Proc. QPL (2020).

\bibitem{marco2023qverif}
Lewis, M., Soudjani, S., Zuliani, P.: Formal verification of quantum programs:
  Theory, tools, and challenges. ACM TQC
  \textbf{5}(1) (2023).

\bibitem{li2022qasmbenchlowlevelqasmbenchmark}
Li, A., Stein, S., Krishnamoorthy, S., Ang, J.: Qasmbench: A low-level qasm
  benchmark suite for NISQ evaluation and simulation. arXiv:2005.13018 (2022).

\bibitem{qcecDoc}
{MQT}-{QCEC},
  \url{https://mqt.readthedocs.io/projects/qcec/en/latest/}

\bibitem{nielsenChuang2002}
Nielsen, M.A., Chuang, I., Grover, L.K.: Quantum computation and quantum
  information. Am. J. Ph.  \textbf{70}(5),  558--559 (2002).

\bibitem{Peham_2022}
Peham, T., Burgholzer, L., Wille, R.: Equivalence checking of quantum circuits
  with the ZX-calculus. IEEE JESTCS \textbf{12}(3),  662–675 (2022).
  
\bibitem{pyzxDoc}
Py{ZX}.
  \url{https://pyzx.readthedocs.io/en/latest/gettingstarted.html}

\bibitem{sander2024equivalencecheckingquantumcircuits}
Sander, A., Burgholzer, L., Wille, R.: Equivalence checking of quantum circuits
  via intermediary matrix product operator. arXiv:2410.10946 (2024).

\bibitem{Shor_1997}
\shor, P.W.: Polynomial-time algorithms for prime factorization and discrete
  logarithms on a quantum computer. SIAM J. Comp.  \textbf{26}(5),
  1484–1509 (1997). 

\bibitem{veriqcBench}
Veri{QC} dynamic quantum circuits benchmarks,
  \url{https://github.com/Veriqc/EC-for-Dynamic-Quantum-Circuits/tree/main/Benchmarks2}

\bibitem{Vilmart2020SOP}
Vilmart, R.: The structure of sum-over-paths, its consequences, and
  completeness for Clifford. arXiv:2003.05678 (2020).

\bibitem{vilmart2023rewriting}
Vilmart, R.: Rewriting and completeness of sum-over-paths in dyadic fragments
  of quantum computing. LMCS \textbf{20(1)} (2024).

\bibitem{wei2022sliqec}
Wei, C.Y. \textit{et al.}: Accurate BDD-based unitary
  operator manipulation for scalable and robust quantum circuit verification.
  In: DAC (2022).
\end{thebibliography}
\end{document}